\documentclass[12pt]{article}
\usepackage[a4paper, total={6in, 9.5in}]{geometry}
\usepackage{amsthm}
\newtheorem{Assumption}{Assumption}
\newtheorem{Theorem}{Theorem}
\usepackage{amsmath}
\usepackage{amsfonts}
\usepackage{authblk}
\usepackage{bm}
\usepackage{bbm}
\usepackage{amsthm}
\usepackage{comment}
\usepackage[nottoc]{tocbibind}
\usepackage{graphicx}
\usepackage{booktabs}
\usepackage{subcaption}
\usepackage{changepage}
\usepackage{natbib}
\usepackage{abstract}
\usepackage{lineno}
\usepackage{multirow}
\usepackage[colorlinks=true,citecolor=blue,runcolor=blue,linkcolor=blue, pdfborder={0 0 0}]{hyperref}
\title{Inference on multiple quantiles in regression models by a rank-score approach}
\usepackage{xcolor}

\newcommand{\keywords}[1]{%
  \par\vspace{.5em}\noindent\textbf{Keywords — }#1\par
}
\author[1]{Riccardo De Santis\footnote{Corresponding author:  \textit{riccardo.desantis@unipd.it}}}
\author[2]{Anna Vesely}
\author[3]{Angela Andreella}

\affil[1]{University of Padova, Department of Statistical Sciences, Padova, Italy} 
\affil[2] {University of Bologna, Department of Statistics, Bologna, Italy} 
\affil[3]{Ca' Foscari University of Venice, Department of Economics, Venice, Italy} 
\date{}

\begin{document}
\maketitle

\begin{abstract}
This paper tackles the challenge of performing multiple quantile regressions across different quantile levels and the associated problem of controlling the familywise error rate, an issue that is generally overlooked in practice. %\st{In many applied fields, it is common practice to test each quantile separately, resulting in a multiplicity issue that is often overlooked.}
We propose a multivariate extension of the rank-score test and embed it within a closed-testing procedure to efficiently account for multiple testing. Then we further generalize the multivariate test to enhance statistical power against alternatives in selected directions. Theoretical foundations and simulation studies demonstrate that our method effectively controls the familywise error rate while achieving higher power than traditional corrections, such as Bonferroni. %\st{Overall, our results highlight the importance of applying rigorous statistical procedures in quantile regression analyses to avoid false discoveries and strengthen the reliability of empirical conclusions.}
\end{abstract}


\section{Introduction}

Quantile regression extends linear regression by estimating a conditional quantile of a response of interest, providing a more comprehensive view of the relationship between covariates and the response beyond just the conditional mean. Fitting multiple quantile regressions at different quantile levels is crucial in many areas where the impact of some covariates varies across the distribution. For example, in finance, investment firms study the effect of macroeconomic variables on the lower quantiles of stock returns during economic downturns, focusing on worst-case scenarios for risk management and loss mitigation. In medical research, the effect of a treatment intended to reduce cholesterol may be more pronounced at the upper quantiles of the cholesterol distribution, particularly for patients with extremely high initial levels. Quantile regressions help identify which patients are most likely to benefit from the treatment, enhancing the effectiveness of medical interventions. In such contexts, analyzing a single quantile value alone could overlook significant results of interest.

Motivated by these considerations, empirical analyses routinely estimate quantile regressions at multiple, pre-specified quantile levels $\tau$ \citep{yu2003quantile}.
Common choices include reference intervals in medicine, such as $\tau \in \{0.025, 0.975\}$ \citep{jiang2024estimating}, or a set of clinically relevant centiles commonly used in pediatric growth chart construction, e.g., $\tau \in \{0.03, 0.10, 0.50, 0.90, 0.97\}$ \citep{kiserud2018world}. %\textcolor{red}{qua si possono mettere anche le ref che vi ho messo per il dataset, ma ho cercato di mettere esempi che siano piu o meno noti}. 
In economics, empirical applications often focus on quartiles or deciles of the conditional distribution, whereas in finance the emphasis is typically on lower extreme quantiles, such as $\tau \in \{0.05, 0.10\}$, to capture tail risk \citep{ul2025distributional}. 
Conversely, studies in climatology and ecology frequently target upper extreme quantiles, e.g., $\tau \in \{0.95, 0.9\}$, to analyze extreme high outcomes or limiting relationships \citep{mazzoglio2025mapping, cade2003gentle}. 
These quantile levels are usually selected a priori to address specific scientific or policy-relevant questions, rather than to provide a uniform characterization of the entire conditional distribution.

This practice, however, leads to a multiple testing problem~\citep{goeman2014multiple}, where performing $K$ separate tests increases the risk of false discoveries if $p$-values are not adjusted to control family errors such as familywise error (FWER) or false discovery rates. Although this issue is partially well understood from a statistical perspective, it is frequently overlooked in empirical applications: researchers often apply quantile regression across multiple quantiles and draw inferences on covariate effects without explicitly addressing multiple testing concerns. Numerous examples can be found in studies spanning markets~\citep{wang2012computing}, biology~\citep{agarwal2019quantile}, management~\citep{bien2016study,karimi2024r}, economics~\citep{fousekis2005demand,falk2012quantile}, public health~\citep{juvanhol2016factors}, health~\citep{cheah2022malaysian}, education~\citep{bassett2002quantile}, energy~\citep{cheng2020quantile}, and climate~\citep{chamaille2007detecting}. This widespread practice highlights a methodological gap, emphasizing the need for the field to develop and apply rigorous statistical methods that properly account for multiple comparisons and ensure the validity of the findings~\citep{ioannidis2005most}.

Within the quantile regression literature, the multiplicity issue is rarely addressed explicitly and, when it is, researchers typically rely on conservative adjustments such as Bonferroni or Holm corrections. To the best of our knowledge, only \cite{mrkvicka2023globalquantileregression} propose a multiple testing procedure specifically tailored to this framework; however, their approach has some theoretical limitations and demonstrates shortcomings regarding type I error control in their simulation study.
We thus focus on the problem of testing the evidence of a covariate effect across multiple quantile levels, where the levels form a finite set of quantiles $\{\tau_1, \dots, \tau_K\}$.

In contrast, a different approach based on stochastic processes focuses on regions of quantiles over a continuous subset of the whole interval $[0,1]$ \citep{sun2021model,park2017hypothesis, belloni2019conditional} or based on distributinal regression models \citep{klein2024distributional}. The reader might wonder whether the adoption of this functional approach, i.e., the analysis of the entire conditional distribution (or of continuous sub-regions) at once, is always preferable when doing inference at a finite discrete number of quantiles. This is not necessarily the case, but the appropriate inferential strategy depends on the specific scientific question. Functional approaches often rely on simultaneous confidence bands constructed from the supremum of a quantile process over a region of quantiles, which is approximated by a discretization of the problem through a grid of values or resampling-based procedures. As the width of this region increases, such procedures can become overly conservative, and post-hoc inference on subregions leads to double-dipping issues. When interest is explicitly focused on a pre-specified set of quantiles, as in the motivating examples above, procedures that optimally control the FWER across these quantiles are then more appropriate. %\textcolor{red}{Analogous considerations apply to distributional regression models \citep{klein2024distributional}: when only a few quantiles are relevant, targeting them directly is more interpretable and statistically efficient than modeling the full distribution.}

We then propose a principled yet straightforward framework for performing valid inference across multiple quantile levels. A naive but valid approach would be to apply standard corrections, such as Bonferroni, to the set of $p$-values obtained from the $K$ separate quantile regressions. However, as demonstrated by~\cite{goeman2021only}, only closed testing procedures are admissible for strong FWER control; every other method is either equivalent to, or uniformly dominated by, some closed testing procedure. Consequently, we adopt the closed testing principle as the inferential backbone of our approach. The closed testing framework controls the FWER by testing all possible intersections of the individual hypotheses $H_0^1, \dots, H_0^K$ using valid local tests. In short, for any $j = 1,\dots, K$, we can reject $H_0^j$ at a pre-determined level $\alpha$ if all intersection hypotheses containing $H_0^j$ are rejected at that level $\alpha$. The validity and efficiency of the closed testing procedure thus hinge on the choice of these local tests. The commonly used Wald-type test for quantile regression~\citep{koenker1996quantile} suffers from finite-sample distortions, especially when heteroscedasticity affects the data, and slow convergence to its asymptotic distribution, even under independent and identically distributed errors~\citep{de2025closed, koenker2005quantile}.

A robust alternative is the rank-score test~\citep{gutenbrunner1993tests, koenker1999goodness, kocherginsky2005practical}, whose multivariate version, however, has not yet been fully developed and theoretically investigated. We therefore first propose the multivariate extension of the rank-score test and study its theoretical proprieties. We then develop a generalization of the multivariate rank-score test, based on alternative weighting matrices, and integrate the resulting class of statistics as the local test within the closed testing procedure. This approach increases analytical flexibility and guarantees valid control of type I errors across multiple quantile levels, thereby improving the reliability of statistical inferences from quantile regression analyses.

To facilitate the application of the proposed methodology, we provide an \texttt{R} package named \texttt{quasar} (QUAntile Selective inference Adjustment via Rank-scores) available on CRAN at \texttt{https://CRAN.R-project.org/package=quasar}.

The manuscript is organized as follows. Section~\ref{setting} presents the necessary background on quantile regression, while Section~\ref{closed_testing} introduces the closed testing procedure. Section~\ref{rank_score} details the theory of the multivariate version of the rank-score test. Section~\ref{gen_rank_test} extends the multivariate rank-score test by introducing a class of generalized test statistics based on alternative weighting matrices. Section~\ref{sim} presents simulation studies showing that the Wald-type test fails to control the type I error rate, whereas the proposed method effectively controls the FWER while maintaining competitive power compared to traditional corrections such as Bonferroni. Finally, Section~\ref{discussion} concludes the paper by summarizing the main findings and implications of our work.

\section{Setting}\label{setting}

Consider a sample of $n$ independent units. For each $i=1,\dots,n$, we observe a continuous response $y_{ni}\in\mathbb{R}$, a scalar covariate of interest
$x_{ni}\in\mathbb{R}$, and a $p$-dimensional vector of nuisance covariates $z_{ni}\in\mathbb{R}^p$, where $p = 1, \dots, n-2$. The first element of $z_{ni}$ equals $1$ for all units and acts as an intercept. Throughout, we make explicit the dependence on $n$ whenever relevant. 

We assume that the following location-scale model~\citep{koenker1999goodness} generates the data:
\begin{equation*}
Y_{ni}|x_{ni},z_{ni}=x_{ni}\beta+z_{ni}^\top\gamma+(x_{ni}\theta+z_{ni}^\top\eta)U_{ni},
\end{equation*}
where $\beta,\theta\in\mathbb{R}$, $\gamma,\eta\in\mathbb{R}^p$, and $U_{ni}$ are independent and identically distributed zero-mean errors from a distribution function $F$. Hence, $Y_{ni}$ is a random variable whose distribution depends on $F$, while $y_{ni}$ represents its realization. Further, denote $\Gamma_n$ as a diagonal $n\times n$ matrix with $i$th entry $\sigma_{ni}=|x_{ni}\theta+z_{ni}\eta|$; this specification allows for relaxation over the homoscedasticity assumption of the units. The only requirement we set is $\eta_1 \ne 0$.

We remark here that the location-scale specification is a standard and widely used benchmark in the quantile regression literature, as it delivers a tractable representation of conditional quantiles while allowing for heteroskedasticity; see, e.g., 
\cite{koenker2005quantile, koenker2010rank, koenker1999goodness, koenker1996quantile}.

For a given level $0<\tau <1$, the conditional \(\tau\)th quantile function of the response \( {Y}_{ni} \), given \( {x}_{ni} \) and \( {z}_{ni} \), is
\begin{equation*}
    Q_{Y_{ni} \mid x_{ni},z_{ni}}(\tau \mid x_{ni}, z_{ni}) = \beta(\tau) x_{ni} + \gamma(\tau) z_{ni},
\end{equation*}
where \( \beta(\tau)=\beta+\theta F^{-1}(\tau) \) is the quantile-specific coefficient for \( {x}_{ni} \), and \( \gamma(\tau)=\gamma+\eta F^{-1}(\tau) \) for \( {z}_{ni} \). All results in the manuscript rely on, sometimes implicit, conditioning on the observed covariates that define the design matrix, as is customary in regression-type models.
%All the results of the manuscript will rely on the (sometimes implicit) conditioning over the observed covariates, which define the design matrix; eventually, this is usual in regression-type models.

The purpose of this manuscript is based on the following problem: given a sample of $n$ elements, we want to test the effect of the target covariate at $K$ different quantile levels $\tau_1,\ldots,\tau_K$. Accordingly, we formulate the following set of $K$ null hypotheses, each tested against a two-sided alternative:
%\begin{align}
%    H_0^1 &: \beta(\tau_1) = 0 \nonumber \\ 
%    & \vdots \label{eq:null} \\ 
%    H_0^k &: \beta(\tau_k) = 0 \nonumber
%\end{align}
\begin{equation*}%\label{eq:null}
\begin{cases}
H_0^1: \beta(\tau_1) = 0\\
\;\vdots\\
H_0^K: \beta(\tau_K) = 0.
\end{cases}
%    H_0^1 &: \beta(\tau_1) = 0 \nonumber \\ 
%    & \vdots  \\ 
%    H_0^k &: \beta(\tau_k) = 0 \nonumber
\end{equation*}
For simplicity, we will refer to these null hypotheses simply as hypotheses throughout the paper.

Although testing for the nullity of the coefficients is the most common practice when investigating statistical hypotheses, this manuscript will actually allow for any null values $\beta_0(\tau_1),\dots,\beta_0(\tau_K)$ of interest. Furthermore, one-sided alternatives can be easily derived, as well.

We introduce some minimal assumptions required for the remainder of the paper. These are not meant to limit practical applicability but to exclude cumbersome or degenerate cases. They are inspired by~\citet{gutenbrunner1992regression} and~\citet{koenker1999goodness} and are standard in inference procedures for quantile regressions.

\begin{Assumption}\label{ass:density}
    $F$ has continuous density $f$ with respect to the Lebesgue measure, and $f$ is positive and finite on \{$t: 0 < F(t) < 1 $\}.
\end{Assumption}

\begin{Assumption}\label{ass:intercept}
    The design matrix $M_n=\begin{bmatrix}
        X_n & Z_n
    \end{bmatrix}$ of dimension $n \times (p+1)$ contains an intercept.
\end{Assumption}

\begin{Assumption}\label{ass:design1}
    $||M_n||_{\infty}=o(n^{1/2})$, where  $||\cdot||_{\infty}=\max_{j = 1, \ldots, p+1}\sum_{i=1}^{n} |m_{ij}|$.
\end{Assumption}

\begin{Assumption}\label{ass:design2}
    $Q_n=n^{-1} M^\top_n  M_n \xrightarrow{n \xrightarrow{}\infty} Q$, where  $Q$
     is a $(p+1) \times (p+1)$  matrix. 
\end{Assumption}

\begin{Assumption}\label{ass:design3}
    For $n=1,2,\dots$ and $i \le n$, there exist two finite constants $\lambda_1, \lambda_2 >0$ such that $\lambda_1 < \sigma_{ni} < \lambda_2$. 
\end{Assumption}

\begin{Assumption}\label{ass:design4}
    $C_n=n^{-1} M^\top_n \Gamma^{-1}_n M_n \xrightarrow{n \xrightarrow{}\infty} C$, where $C$
     is a $(p+1) \times (p+1)$ matrix.
\end{Assumption}

\begin{Assumption}\label{ass:design5}
    $V_{n}=n^{-1} (X_{n}-\hat{X}_{n})^\top  (X_{n}-\hat{X}_{n}) \xrightarrow{n \xrightarrow{}\infty} V$, where $V$ is a scalar and $\hat{X}_n=Z_n(Z_n^\top \Gamma_n^{-1} Z_n)^{-1}Z^\top_n\Gamma^{-1}_n X_n$.
\end{Assumption}

Assumption~\ref{ass:density} formalizes that we restrict our attention to absolutely continuous response variables. Assumption~\ref{ass:intercept} asks that an intercept must always be included in the fitted model. Assumption~\ref{ass:design1} requires that each entry of the design matrix does not grow too fast; it is weaker than assuming finite elements over the entire design matrix. Assumption~\ref{ass:design3} posits that the variances of the response vector are bounded away from zero and infinity. Assumptions~\ref{ass:design2}-\ref{ass:design4}-\ref{ass:design5} guarantee the existence of well-defined limits for some matrices based on the variance-covariance matrix of the covariates.

With these conditions in place, the following section outlines the closed-testing procedure used to achieve strong control of the FWER.

\section{Closed testing principle}\label{closed_testing}
To test the hypotheses of interest $H_0^1, \ldots, H_0^K$ with FWER control at a pre-specified level $\alpha\in (0,1)$, we employ the closed testing principle~\citep{marcus1976closed}. This is the optimal framework for controlling the FWER, as any valid multiple testing method is either equivalent to a closed testing procedure or uniformly improved by one. This optimality property extends beyond FWER control and holds for other error measures, such as the false discovery proportion~\citep{goeman2021only}.

Let $\mathcal{K}=\{1,\ldots,K\}$ denote the index set of all hypotheses, and let $\mathcal{K}_0\subseteq\mathcal{K}$ represent the unknown subset of true hypotheses. The closed testing principle requires considering not only individual hypotheses, but also intersection hypotheses of the form
\[H_0^{\mathcal{C}}=\bigcap_{j\in\mathcal{C}} H_0^j,\qquad\mathcal{C}\subseteq\mathcal{K}.\]
Such an intersection hypothesis is true if and only if all individual hypotheses involved are true, i.e., $\mathcal{C}\subseteq\mathcal{K}_0$. Furthermore, the closed testing procedure requires a valid $\alpha$-level local test
$\psi : 2^{\mathcal{K}}\rightarrow\{0,1\}$, where $2^{\mathcal{K}}$ is the power set of $\mathcal{K}$ and $\psi(\mathcal{C})=1$ denotes rejection of $H_0^{\mathcal{C}}$ for $\mathcal{C}\subseteq\mathcal{K}$. The local test is said to be at level $\alpha$ if $\text{pr}(\psi(\mathcal{C}) = 1)\le \alpha$ for any $\mathcal{C} \subseteq\mathcal{K}_0$, where the probability is taken under the true, unknown data-generating process. By convention, the intersection hypothesis $H_0^{\emptyset}$ corresponding to the empty set is always true, and $\psi(\emptyset) = 0$.

Based on the local test, closed testing defines a global decision rule 
\[\bar{\psi} : 2^{\mathcal{K}}\rightarrow\{0,1\},\qquad \bar{\psi}=\min\{\psi(\mathcal{E}) : \mathcal{C}\subseteq\mathcal{E}\subseteq\mathcal{K}\}.\]
Thus, closed testing rejects a hypothesis $H_0^{\mathcal{C}}$ whenever the local test rejects all intersection hypotheses corresponding to supersets of $\mathcal{C}$. By construction, the procedure provides strong control of the FWER:
\[\text{pr}(\exists\,\mathcal{C}\subseteq \mathcal{K}_0 : \bar{\psi}(\mathcal{C})=1)\leq\alpha.\]
While the procedure applies to all intersection hypotheses, our focus here is, as is typically the case, on testing the individual hypotheses only.

For illustration, consider the following simple example, adapted from~\citet{goeman2011multiple}. Figure~\ref{fig1} shows the tree of all intersection hypotheses for $K=3$, reflecting the subset-superset relationships. Each level corresponds to index sets of the same size, and an arrow from a node $\mathcal{E}$ to a node $\mathcal{C}$ indicates that $\mathcal{C}\subset\mathcal{E}$. The closed testing procedure rejects an individual hypothesis, for example $H_1$, only if the local test rejects all intersection hypotheses that include it, namely $H_{\{1\}}=H_1$, $H_{\{1,2\}}=H_1\cap H_2$, $H_{\{1,3\}}=H_1\cap H_3$, and $H_{\{1,2,3\}}=H_1\cap H_2\cap H_3$.

\begin{figure}
\centering
\includegraphics[width=.5\textwidth]{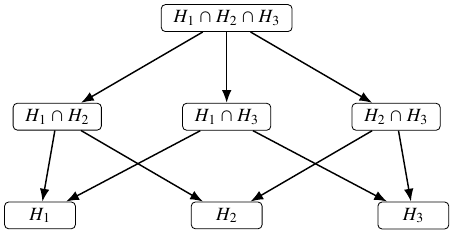}
\caption{Tree of intersection hypotheses for three individual hypotheses. Adapted from \citet{goeman2011multiple}.}
\label{fig1}
\end{figure}

In the context of quantile regression, the possible local tests are limited to either tests for individual hypotheses, such as the univariate rank-score test~\citep[Section~3.5.4]{koenker2005quantile}, combined with a Bonferroni correction, or the Wald-type test of~\citet[Section~3.5.3]{koenker2005quantile}. However, the Bonferroni correction is known to have notoriously low statistical power~\citep{pesarin2001multivariate}, while the Wald-type alternative converges slowly to its asymptotic distribution and exhibits distortions in finite samples. These limitations are illustrated, for example, in Section~\ref{sim} and in~\citet{de2025closed}.
%, where we explored an initial proposal for a closed testing procedure that relied on this approach.

The next section introduces a local test that constitutes the multivariate extension of the test based on regression rank-scores and overcomes the small-sample limitations of the Wald-type approach.

\section{Multivariate rank-score test}\label{rank_score}
The proposed local test relies on a rank-based testing approach. For the hypothesis $H_0^j$ corresponding to a single quantile of interest, the univariate test based on regression rank-scores is described in~\citet[Section~3.5.4]{koenker2005quantile}. This is a score-type test, meaning that the test statistic is derived from the sub-model implied by $H_0^j$. To our knowledge, the literature has focused exclusively on such univariate hypotheses. This section outlines the theoretical proposal of this manuscript, that is, the extension to intersection hypotheses over $k = 1, \dots, K$ quantiles, where the case $k=1$ coincides with the cited univariate test. For simplicity of notation, we will refer to the intersection hypothesis corresponding to the first $k$ quantiles $\{\tau_1,\ldots,\tau_k\}$, that is,
\begin{equation*}
H_0^{\{1,\ldots,k\}}=H_0^1\cap\ldots\cap H_0^k,\qquad H_0^j:\beta(\tau_j)=0.
\end{equation*}
%This section outlines the theoretical proposal of this manuscript, that is, the extension to intersection hypotheses. Indeed, according to our knowledge, the literature focuses only on univariate hypotheses. Therefore, we will build a theory for intersection hypotheses about $k$ quantiles, which includes the special case $k=1$. Hence, our null hypothesis is of the form $H_0^1\cap \dots \cap H_0^k$.

The first key component of the proposed approach is the vector of regression rank-scores, derived from the estimation procedure of the quantile regression model~\citep[Section~6.2]{koenker2005quantile} and computed under the sub-model implied by the intersection null hypothesis, to which we will refer as the null model. Let $0<\tau<1$ denote a generic quantile level. The regression rank-scores are
%; since we are considering the nullity of the target parameters, all the $\beta(\tau)$ are set equal to zero. In such cases, the regression rank-scores are defined as
\begin{equation*}
\hat{a}_n(\tau) = \arg\max_{a \in [0, 1]^n} \{a^\top Y_{n} : \; Z^\top_n a = (1 - \tau )Z^\top _n1_n\},    
\end{equation*}
where $Y_{n}$ denotes the $n$-dimensional vector of the responses and $1_n$ is an $n$-dimensional vector of ones. The $n$-vector $\hat{a}_n(\tau)$ represents the dual solution to the quantile regression problem
\begin{equation*}
\hat{\gamma}(\tau)= \arg\min_{\gamma \in \mathbb{R}^{p}}\sum_{i=1}^n \rho_\tau (y_{ni}-z_{ni}^\top\gamma),    
\end{equation*}
where $\rho_\tau(u)=u(\tau-\mathbf{1}\{u<0\})$ denotes the quantile loss function, and $\mathbf{1}\{\cdot\}$ is the indicator function. Notably, the regression rank-scores are $\tau$-dependent as explicitly stated.

As detailed by \citet{koenker2005quantile}, the score function is obtained by integrating an appropriate score-generating function $\phi:(0,1) \to \mathbb{R}$, which must be non-decreasing and square-integrable. Letting $\hat{a}_{ni}(\tau)$ denote the $i$th component of $\hat{a}_{n}(\tau)$, the score function is thus defined as
\begin{equation*}
    \hat{b}_{ni}(\tau)=\int^1_0 \phi(u) \, d\hat{a}_{ni}(u).
\end{equation*}

While \citet{gutenbrunner1993tests} introduced general score-generating functions $\phi(\cdot)$ for shift-location problems, \citet{koenker1999goodness} advocated the use of a quantile-specific function for $\tau$-specific problems:
\begin{equation*}
    \phi_\tau(t)=\tau-\mathbf{1}\{t < \tau\},
\end{equation*}
which yields
\begin{equation*}
    \hat{b}_{ni}(\tau)=\hat{a}_{ni}(\tau)-(1-\tau).
\end{equation*}

%We now introduce the elements $d_{ni}$, which define the entries of the vector $D_n=X_n-\hat{X}_n$.
Let $d_{ni}$ denote the $i$th entry of the vector $D_n=X_n-\hat{X}_n$. Following \citet{gutenbrunner1992regression}, the quantile rank-score process is
\begin{equation}\label{eq:hatprocess}
    \hat{W}^d_n=\left\{\hat{W}_n^d(\tau)=n^{-1/2}\sum_{i=1}^n d_{ni} [\hat{a}_{ni}(\tau)-(1-\tau)], \; 0< \tau <1\right\}.
\end{equation}
This stochastic process forms the basis of the score-based test statistic proposed below. It is characterized by the following theorem from \citet[Theorem~1(iii)]{gutenbrunner1992regression}.

\begin{Theorem}\label{theo:process}
    Let Assumptions~1--7 hold and suppose that $H_0^{\{1,\ldots,k\}}$ is true. Then
    \[\hat{W}^d_n \xrightarrow{d} V^{1/2}W^*,\]
    where $W^*$ is a Brownian bridge. The convergence in distribution of the process is understood in terms of the Prokhorov topology on the Skorokhod space $\mathbb{D}(0,1)$.
\end{Theorem}

The original theorem is stated for general triangular arrays of possibly multidimensional vectors, with an unspecified structure in place of the elements $d_{ni}$.
Theorem \ref{theo:process} presents a more restrictive formulation, sufficient for the purposes of this paper, yet flexible enough to allow extensions, for example, to inference on a multivariate target coefficient.
%Here, we report a more restrictive version, which is sufficient for the purposes of this paper, while still allowing extensions, for instance, to inference on a multivariate target coefficient.

The distribution of the score vector over multiple quantiles is obtained as follows. Define
\begin{equation*}
    S_n=n^{-1/2}(X_n-\hat{X}_n)^\top  \hat{b}_n
\end{equation*}
and the $k \times k$ matrix $\Delta$ whose $(\ell,r)$th entry is $\delta_{\ell r}=\,\min\{\tau_\ell,\tau_r\}-\tau_\ell\tau_r$. The test statistic for the local test is
\begin{equation}
    T_n=S_n^\top A^{-1}_nS_n, \label{eq:tn}
\end{equation}
where $A_n = V_n\Delta$, and $V_n$ is the scalar introduced in Assumption~\ref{ass:design5}.
%is a $k \times k$ matrix whose $(\ell,r)$th entry equals $V_n \delta_{\ell r}$ with $\delta_{\ell r}=\,\min\{\tau_\ell,\tau_r\}-\tau_\ell\tau_r$.

\begin{Theorem}\label{theo:tstatnull}
    Let Assumptions~1--7 hold and suppose that $H_0^{\{1,\ldots,k\}}$ is true. Then $T_n\xrightarrow{d} T \sim \chi^2_k$, where $\chi^2_k$ denotes a chi-squared distribution with $k$ degrees of freedom.
\end{Theorem}

\begin{proof}
The result follows directly from properties of the Brownian bridge. For two points $\ell,r$, such that $\ell<r$, we have
\begin{equation*}
    \begin{bmatrix}
W^*(\ell) \\ W^*(r) 
\end{bmatrix} \sim \mathcal{N}\left(\begin{bmatrix}
0 \\ 0
\end{bmatrix}, 
\begin{bmatrix}
\ell(1-\ell) & \ell(1-r) \\
\ell(1-r) & r(1-r)
\end{bmatrix}\right)
\end{equation*}
where $\mathcal{N}$ denotes the multivariate normal distribution. Multiplying the Brownian bridge by a constant $c$ scales each element of the covariance matrix by $c^2$. By Theorem~\ref{theo:process} and Slutsky's theorem, it follows that, for $k$ points in $(0,1)$, the limiting distribution of $S_n$ is $k$-variate normal with zero mean vector and covariance matrix $A$, where $A$ is a $k \times k$ matrix whose $(\ell,r)$th entry equals $V \delta_{\ell r}$. The final claim follows from the relationship between the normal and chi-squared distributions.
\end{proof}

Theorem~\ref{theo:tstatnull} enables rank-score-based testing of hypotheses across multiple quantiles simultaneously. Combined with the closed-testing framework described in Section~\ref{closed_testing}, it allows inference on individual quantiles while properly adjusting for multiplicity. Finally, the extension to the multivariate case, where $X_n$ has multiple columns, is straightforward, albeit involving more tedious matrix computations.

After establishing the null distribution of the test statistic, we characterize its behavior under alternatives. In particular, its distribution can be derived under local alternatives of the form
\begin{equation*}
H_n^{\{1,\ldots,k\}}=H_n^1\cap\ldots\cap H_n^k,\qquad H_n^j:\beta(\tau_j) = {n}^{-1/2}h(\tau_j),
\end{equation*}
%(still assuming $\beta_0=0$) of the form
%\begin{equation*}%\label{eq:alternative}
%\begin{cases}
%H_n^1: \beta(\tau_1) = h(\tau_1)\,{n}^{-1/2}\\[4pt]
%\;\vdots\\[4pt]
%H_n^k: \beta(\tau_k) = h(\tau_k)\,{n}^{-1/2}.
%\end{cases}
%\end{equation*}
where $h(\cdot)>0$ is a finite constant. In such cases, we obtain the following result. 
\begin{Theorem}\label{theo:alt}
    Let Assumptions~1--7 hold and suppose that $H_n^{\{1,\ldots,k\}}$ is true. Then $T_n\xrightarrow{d} T \sim \chi^2_k(\zeta)$, where $\chi^2_k(\zeta)$ denotes a noncentral chi-squared distribution with $k$ degrees of freedom. The noncentrality parameter is $\zeta=g^\top A^{-1} g$, where $g\in\mathbb{R}^k$ is a vector with components 
    \[g_j=C^{11}f \circ F^{-1}(\tau_j)\, h(\tau_j),\quad j =1,\ldots,k.\]
    Here $f \circ F(\cdot)=f(F(\cdot))$, and $C^{11}=(C^{-1})_{11}$ is the $(1,1)$ element of the inverse of $C$.
\end{Theorem}

\begin{proof}
    Let $\hat{W}^d_n$ be defined as in \eqref{eq:hatprocess}, and
\begin{equation*}
    {W}^d_n=\left\{{W}_n^d(\tau)=n^{-1/2}\sum_{i=1}^n d_{ni} [a^*_{ni}(\tau)-(1-\tau)], \; 0< \tau <1\right\}
\end{equation*}
where $a^*_{ni}(\tau)=\mathbf{1}\{U_{ni}>F^{-1}(\tau)\}$.

The proof of Theorem~1 in \citet{gutenbrunner1992regression} shows that
\begin{equation}\label{eq:difference}
\hat{W}_n^d-W_n^d = -n^{1/2}\left(\hat{G}^d_n(\tau)-G^d_n(\tau)\right) + o_p(1),
\end{equation}
where
\begin{equation*}
    \hat{G}^d_n(\tau)=n^{-1}\sum_{i=1}^n d_{ni} \mathbf{1}\{U_i \le  m_{ni}^\top \hat{t}_n(\tau)/\sigma_{ni}\}
\end{equation*}
and 
\begin{equation*}
    G^d_n(\tau)=n^{-1}\sum_{i=1}^n d_{ni} \mathbf{1}\{U_i \le  m_{ni}^\top t(\tau)/\sigma_{ni}\}
\end{equation*}
are empirical processes, whose definitions refer to the full model. Here 
\[
\hat{t}_n(\tau) =
\begin{bmatrix}
\hat{\beta}_n(\tau) - \beta \\[2pt]
\hat{\gamma}_n(\tau) - \gamma
\end{bmatrix},
\qquad
t(\tau) =
\begin{bmatrix}
\beta(\tau) - \beta \\[2pt]
\gamma(\tau) - \gamma
\end{bmatrix},
\qquad
m_{ni} =
\begin{bmatrix}
x_{ni} \\[2pt]
z_{ni}
\end{bmatrix}.
\]

Using equation \eqref{eq:difference} and a first-order linearization with respect to the regression parameters of the difference $\hat{G}^d_n(\tau)-{G}^d_n(\tau)$, since $\hat{W}^d_n$ is evaluated under $H_0^{\{1,\ldots,k\}}$ while the true model is given by $H_n^{\{1,\ldots,k\}}$, we obtain
\begin{equation*}
    \hat{W}^d_n-{W}^d_n=n^{-1/2}f\circ F^{-1}(\tau) D_n^\top\Gamma_n^{-1}Z_n \hat{\gamma}_n(\tau)
    -n^{-1/2}f\circ F^{-1}(\tau) D_n^\top\Gamma_n^{-1}M_n
    \begin{bmatrix}
        h(\tau) \, {n}^{-1/2} \\ {\gamma}(\tau)
    \end{bmatrix} + o_p(1)
\end{equation*}
where $M_n=\begin{bmatrix}
        X_n & Z_n
    \end{bmatrix}$.

We compare this linearization under the null model and under local alternatives. Under the null model, the relation
\begin{equation*}
    \hat{W}^d_n-{W}^d_n=n^{-1/2}f\circ F^{-1}(\tau)D_n^\top\Gamma_n^{-1}Z_n \left(\hat{\gamma}_n(\tau)-{\gamma}(\tau)\right) + o_p(1)
\end{equation*}
is used to prove Theorem \ref{theo:process} by \citet{gutenbrunner1992regression}. On the other hand, under the local alternatives, the additional term 
\begin{equation}\label{eq:meanchisq}
    n^{-1}f\circ F^{-1}(\tau) D_n^\top\Gamma_n^{-1}X_n h(\tau)
\end{equation}
is deterministic. By standard matrix algebra we get
\begin{equation*}
n^{-1}D_n^\top\Gamma_n^{-1}X_n=(C^{-1}_n)_{11}=C_n^{11},
\end{equation*}
which is the $(1,1)$ element of the inverse of $C_n$.

Hence, using the same argument as in Theorem~\ref{theo:tstatnull} and Slutsky's theorem once more, the statistic $S_n$ has $k$-variate normal limiting distribution with covariance matrix $A$, as before, and mean vector~\eqref{eq:meanchisq}. Further, by standard properties of the normal distribution, the test statistic $T_n={S_n^\top A^{-1}_n S_n}$ converges in distribution to $T \sim \chi^2_k(\zeta)$, a noncentral chi-squared distribution, as stated in the theorem.
\end{proof}

A few remarks conclude this section. When defining the hypotheses of interest, values different from zero can also be tested. For simplicity, consider a single quantile $\tau_j$ and the hypothesis $H_0^j:\beta(\tau_j)=\beta_0(\tau_j)$, for a fixed value $\beta_0(\tau_j)\in\mathbb{R}$. The hypothesis is tested by specifying a null model with $\beta(\tau_j)=\beta_0(\tau_j)$, and optimizing with respect to the remaining regression parameters. This corresponds to the standard use of an offset term in regression modeling. All theoretical results presented in this section remain valid without modification.

Finally, additional care is required in estimating the elements $\sigma_{ni}$, especially when allowing for heteroscedasticity. This issue is discussed by~\citet{koenker1999goodness} and~\citet[Section~3.4]{koenker2005quantile}, who show that these quantities can be estimated by evaluating the conditional density of the response at the $\tau$th conditional quantile, denoted $f_{ni}(Q_{ni}(\tau|z_i))$. Intuitively, such estimation is challenging in finite samples. In the independent but non-identically distributed setting, we follow~\citet{hendricks1992hierarchical} and~\citet{koenker1999goodness}, estimating the density via the difference quotient 
\begin{equation}\label{eq:sparsity}
\hat{f}_{ni}(Q_{ni}(\tau)|z_{ni})
= \frac{2h_n}{z_{ni}^\top \bigl(\hat{\gamma}(\tau + h_n) - \hat{\gamma}(\tau - h_n)\bigr)}
= \frac{2h_n}{d_{ni}}.
\end{equation}
For the selection of the bandwidth parameter $h_n$, we follow the Hall--Sheather rule~\citep{hall1988distribution} as adopted by~\citet{koenker1999goodness}. Since the denominator in~\eqref{eq:sparsity} may not be positive, $\hat{f}_{ni}(Q_{ni}(\tau)|z_{ni})$ is replaced by 
\begin{equation*}
\hat{f}^+_{ni}(Q_{ni}(\tau)|z_{ni})=\max\left\{0.01,\frac{2h_n}{d_{ni}-\varepsilon}\right\},
\end{equation*}
where $\varepsilon>0$ is a small tolerance parameter. This procedure estimates the elements of $\Gamma_n$ up to a scale factor independent of $i$, which cancels in the expression for $\hat{X}_n$~\citep{koenker1999goodness}. The choice of $0.01$ is made for computational reasons and is practically irrelevant.

\section{Generalized multivariate rank-score test}\label{gen_rank_test}

While the test described in the previous section is appealing for its simple mathematical formulation and reference distribution, its power can be significantly improved without compromising the type I error control.

To begin with, consider the distribution of the test statistic under alternatives, in the setting of Theorem~\ref{theo:alt}. Once the noncentrality parameter is made explicit, the limitations in terms of statistical power become evident. Consider, for example, the case $k=2$; the noncentrality parameter is 
\begin{equation*}
    \zeta=g_1^2a_{11}+2g_1g_2a_{12}+g^2_2a_{22}
\end{equation*}
where $A=V\Delta$ and $a_{\ell r}$ denotes the $(\ell,r)$th entry of the matrix $A^{-1}$. Thus, the power of the test decreases as a function of the correlation among the test statistics. Furthermore, the finite-sample counterpart of $A$ is $A_n=V_n\Delta$, where $\Delta$ is deterministic and depends only on the chosen quantile levels.
%can be written as the product of the scalar $V_n$ and the matrix $\Delta$, with elements $\delta_{\ell r}$, which is deterministic and depends only on the chosen quantile levels.

It is possible to replace the matrix $A_n^{-1}$ in the definition of the test statistic~\eqref{eq:tn} with an alternative weighting matrix. The theory of multivariate score tests with different weighting matrices is developed in \citet{goeman2006testing,goeman2011testing} for mean regression models. Motivated by this framework, we introduce an alternative test statistic
\begin{equation*}
        T_n^*=S_n^\top B S_n,
    \end{equation*}
where $B$ is any symmetric, positive-definite, $k \times k$ matrix. Under the null and alternative hypotheses, the original statistic $T_n$ converges in distribution to a (noncentral) $\chi^2_k$ random variable (Theorems~\ref{theo:tstatnull}-\ref{theo:alt}). By contrast, $T_n^*$ converges to a weighted sum of $k$ independent (noncentral) $\chi^2_1$ random variables:

\begin{Theorem}\label{theo:tstatnull_2}
    Let Assumptions~\ref{ass:density}--\ref{ass:design5} hold and suppose that $H_0^{\{1,\ldots,k\}}$ is true. Then $T_n^*\xrightarrow{d} T^*\sim \sum_{i=1}^k \lambda_i \chi^2_1$, where the weights $\lambda_i$ are the eigenvalues of the product matrix $BA$.
\end{Theorem}

\begin{proof}

    From Theorem~\ref{theo:tstatnull}, $S_n \xrightarrow{d} S \sim N(0,A)$. Hence, we may write $A^{-1/2}S= S_z \sim N(0, I_k)$. Then
    \begin{equation*}
        T^*=S^\top B S =S^\top_z \left(A^{1/2} B A^{1/2}\right) S_z = \sum_{i=1}^k \lambda_i (b_i^\top S_z)^2,
    \end{equation*}
    %we can decompose $A^{1/2} B A^{1/2}$ using the singular value decomposition. Once we have $T^*=\sum_i \lambda_i (b_i^\top S_z)$,
    where $b_i$ denote the eigenvectors of the matrix $A^{1/2} B A^{1/2}$. Notice that, by the rotational invariance of the multivariate normal distribution, and being the eigenvectors orthogonal by definition, the random vector $(b_1^\top S_z,\ldots,b_k^\top S_z)^\top$ has components which are independent standard normal random variables.
    %we use the invariance of the distribution between $b_i^\top S_z$ and $S_{z}$ which in both case is $N(0,1)$.
    The final claim follows from the fact that, given the assumptions made, the eigenvalues of $A^{1/2} B A^{1/2}$ and $BA$ coincide. 
\end{proof}
The reference distribution cannot be derived analytically, but it can be approximated numerically, for instance using the method proposed by \citet{imhof1961computing}. Further, by the definition of $B$, the weights $\lambda_i$ of the distribution are all strictly positive.

While the choice of $B$ does not affect the control of the type I error, it is crucial for the power properties of the test. When testing multiple parameters simultaneously, it is generally not possible to maximize power against all alternatives \citep{goeman2006testing}. Indeed, if the region of alternatives is viewed as a $k$-dimensional ball centered at the null, the intersection test is primarily sensitive along a particular direction. The choice of $B$ effectively rotates this direction, influencing which alternatives the test is most powerful against. The following theorem derives the distribution of the test statistic under local alternatives.

\begin{Theorem}
    Let Assumptions~\ref{ass:density}--\ref{ass:design5} hold and suppose that $H_n^{\{1,\ldots,k\}}$ is true. Then $T_n^*\xrightarrow{d} T^*\sim \sum_{i=1}^k \lambda_i \chi^2_1(\zeta_i)$,
    %where $\chi^2_1([b_i^\top g_z]^2)$ denotes a noncentral chi-squared distribution with $1$ degree of freedom and noncentral parameters given by the elements $[b_i^\top g_z]^2, \; i=1,\dots, k $, and $g_z=A^{-1/2}g$. 
    where $\zeta_i=(b_i^\top g_z)^2$ and $g_z=A^{-1/2}g$.
\end{Theorem}

\begin{proof}
    The proof proceeds along the same lines as that of Theorem~\ref{theo:tstatnull_2}. From Theorem \ref{theo:alt}, $S_n \xrightarrow{d} S \sim N(g,A)$. Then
    \begin{equation*}
    T^{*}=(S_z+g_z)^\top A^{1/2}BA^{1/2} (S_z+g_z) = \sum_{i=1}^k \lambda_i [b_i^\top (S_z+g_z)]^2
    \end{equation*}
    where the elements $\lambda_i$ and $b_i$ are, respectively, the eigenvalues and eigenvectors of the matrix $A^{1/2} B A^{1/2}$. The properties of the normal distribution apply as in Theorem \ref{theo:alt}, since each element $b_i^\top g_z$ applies as a location shift of a standard normal random variable.
    %If we use the singular value decomposition of the matrix $A^{1/2}BA^{1/2}$ we get
    %\begin{equation*}
    %    \sum_{i=1}^k \lambda_i [b_i^\top (S_z+g_z)]^2.
    %\end{equation*}
    %Hence we get that the distribution under local alternatives is $\sum_{i=1}^k \lambda_i [\chi^2_1([b_i^\top g_z]^2)]$, where $\chi^2_1([b_ig_z]^2)$ denotes a noncentral chi-squared distribution with $1$ degree of freedom and noncentrality parameters given by the elements $[b_i^\top g_z]^2, \; i=1,\dots, k $.
\end{proof}

We conclude this section with some practical remarks. The choice of the weighting matrix clearly affects the distribution of the test statistic under alternatives in multiple ways, influencing both the noncentrality parameters and the weights of the chi-square-based reference distribution. Moreover, the convergence rate to the asymptotic distribution depends on the selected quantile levels. Hence, defining an optimal weighting matrix is not straightforward, even when prioritizing specific quantiles, while simultaneously accounting for the multiplicity inherent in the global testing problem.
%it is not trivial to define an optimal choice even if we have preference over some quantiles, while still taking into account the multiplicity problem in a global view.

Consider the identity matrix $I_k$ as the weighting matrix, and $S$ to be a generic zero-mean normal random vector with positive-definite variance. Then $(S-g)^\top(S-g)$ is a weighted Euclidean distance, with weights given by the variance matrix of $S$. By performing a singular value decomposition of this matrix,
\begin{equation*}
(S-g)^\top(S-g) = \sum_{i=1}^k \lambda_i \left(S_{i}+\frac{b_i^\top g}{\sqrt{\lambda_i}}\right)^2.
\end{equation*}
If the variance matrix is equicorrelated with identical diagonal elements, the power is equally distributed across the elements of $g$. In this scenario, the overall power depends on the magnitude of the correlation, as can be seen from the explicit form of the eigenvalues and eigenvectors. In this case, the first eigenvector has identical entries, the sum of the eigenvalues is constant, and the first eigenvalue increases with the correlation. If this eigenvalue dominates the others, the distance induced by $g$ is mainly in the first principal component. The test is primarily sensitive in this direction, leading to increased power. Heterogeneity in either the diagonal or off-diagonal elements of the variance matrix influences both the distribution of power among the components and the overall average power.

Returning to the quantile regression framework, the variance matrix depends only on the quantile levels, up to the scalar factor, and is therefore a deterministic, known quantity. Hence, differences in variances across components affect the distribution of the power among the elements of $g$, penalizing the extreme quantiles. To mitigate this issue, different weighting matrices can be used. This should be viewed as distributing the overall power differently, rather than providing a possible maximization over the entire set of quantiles. For example, to correct for the imbalance induced by heterogeneous variances, one may use the weighting matrix $\mbox{diag}\{\Delta\}^{-1}$. Another natural choice arises when assumptions are made on the distribution of the errors.
%we can try to average the power on the alternatives $h(\tau_j)$. This requires to put the reciprocal of the densities computed at the demanded quantiles in the diagonal of the weighting matrix. 
In this case, the elements of the vector $g$ depend on the error density computed at the target quantiles. To balance power across the alternatives $h(\tau_j)$, the diagonal of the weighting matrix can be set equal to the reciprocals of these densities.
%A wrong choice of the distribution will not average on the alternatives $h(\tau_j)$ but on some more complex combination of the true (unknown) distribution and the one used.
However, if the assumed error distribution is misspecified, this procedure does not properly redistribute power across the true alternatives, but rather across a more complex combination of the true (unknown) distribution and the assumed one.

\section{Simulation study}\label{sim}

For all simulation studies, we set the sample size to $n=100$ and the significance level to $\alpha=0.05$. For readability, the dependence on $n$ is suppressed in the mathematical notation. Covariate pairs $(x_i, z_i)^T$ are drawn from a bivariate normal distribution with mean zero and an equicorrelated covariance matrix with correlation~$0.3$, while different generating mechanisms are considered for the responses $y_i$. Each simulation scenario is replicated $1{,}000$ times. 

First, we define five individual hypotheses, each testing whether the effect of $x_i$ is zero at the quantile levels $\{0.1, 0.25, 0.5, 0.75, 0.9\}$. This results in $2^5 -1=31$ possible intersection hypotheses, excluding the one corresponding to the empty set, as defined in Section~\ref{closed_testing}. We then examine the type I error control of the Wald-type test of~\citet{koenker2005quantile} and the proposed generalized rank-score test, fixing $B = I_5$ (or the appropriate sub-matrix for each hypothesis) as in Section~\ref{gen_rank_test}.
%We simulate $Y_{1}, \dots, Y_{100}$ from a normal distribution with mean $0\,x_i + 0.5\,z_i$ and variance $\sigma^2 = 1$, for $i = 1, \dots, 100$, and repeat this experiment 1000 times. The covariates $(x_i, z_i)^\top$ are drawn from a bivariate normal distribution with mean zero and an equicorrelated covariance matrix with correlation~$0.3$.
In this setting, $y_i \sim \mathcal{N}(0.5+0.5\,z_i, 1)$, so there is no effect of $x_i$ at any quantile level.
%Each intersection hypothesis, as defined in Section~\ref{closed_testing}, is tested, where the individual hypotheses correspond to testing the effect of $x_i$ on the quantiles at level $0.1, 0.25, 0.5, 0.75, 0.9$ and the null value is~0.
The resulting $p$-values, averaged over simulations, are shown in Fig.~\ref{fig:error}.
%displays the corresponding $p$-values, \anna{averaged over simulations}, obtained from the rank-score tests (horizontal axis) and the Wald-type tests (vertical axis).
%The dotted lines indicate the $95\%$ simulation confidence interval for the type~I error.
Only three of the 31 $p$-values fall within the confidence bands for the Wald-type tests, indicating inadequate type~I error control, whereas all generalized rank-score $p$-values lie within the bands, showing proper control. Power is then evaluated by computing the proportion of simulations in which an effect of $x_i$ is detected, and compared between the generalized rank-score tests with $B = A^{-1}$ (Section~\ref{rank_score}) and $B = I_5$ (Section~\ref{gen_rank_test}), or appropriate sub-matrices based on each specific hypothesis. Under $y_i \sim \mathcal{N}(0.5+0.6 x_i + 0.5z_i, 1+ |x_i|)$, Fig.~\ref{fig:ct2} shows that the identity weighting improves power for most intersection hypotheses, with no gains for individual hypotheses along the bisector as expected.

\begin{figure}
\centering
\includegraphics[width=.5\textwidth]{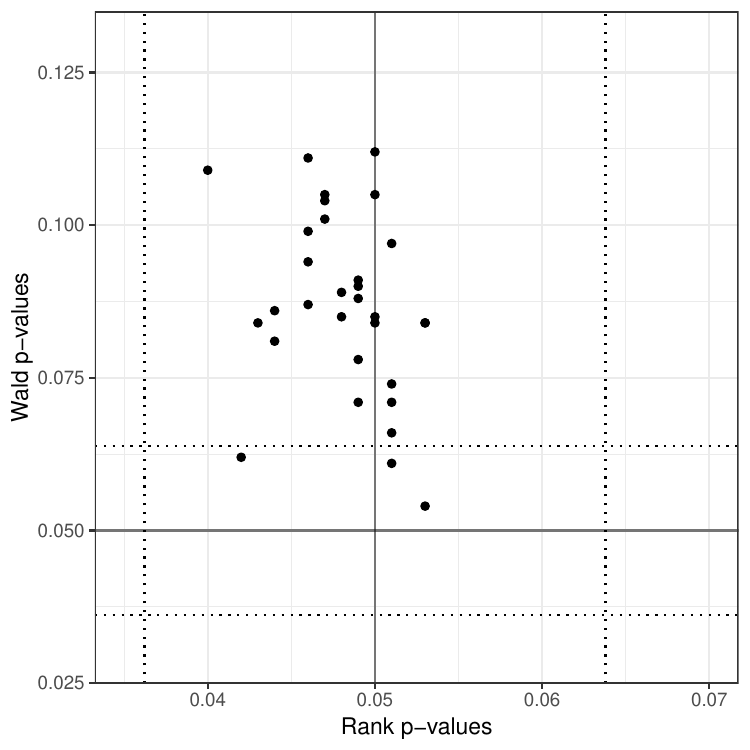}
\caption{Type I error control for Wald-type and generalized rank-score tests with $B = I_5$ for intersection hypotheses. Average $p$-values; solid grey lines correspond to $\alpha$, and dotted lines represent the $95\%$ simulation confidence interval for the empirical type I error.
}
\label{fig:error}
\end{figure}

\begin{figure}
\centering
\includegraphics[width=.5\textwidth]{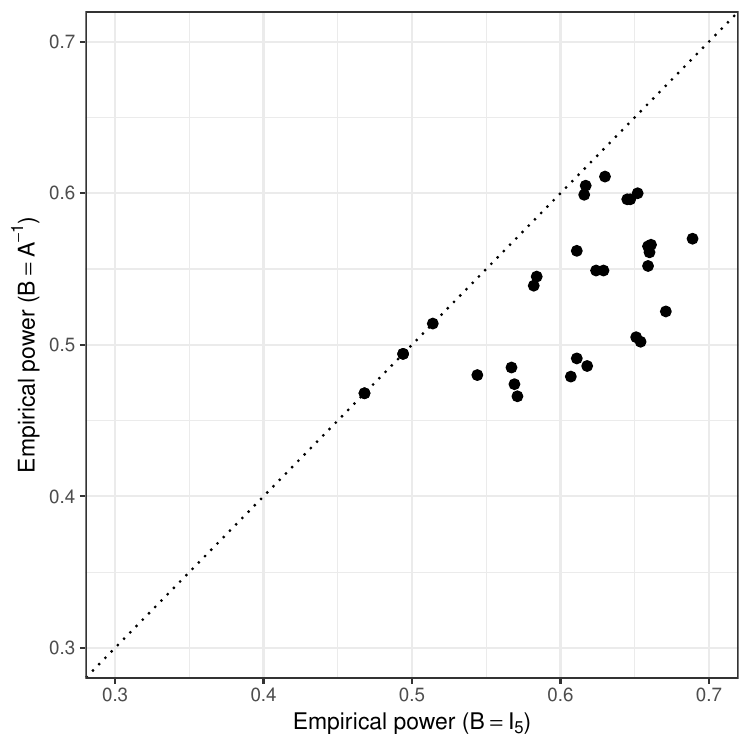}
% note that files may not be rotated
\caption{%Type I error control of rank-score and Wald approaches.
Empirical power of the generalized rank-score tests for the $31$ intersection hypotheses, comparing the identity weighting ($B = I_5$) with the inverse-covariance weighting ($B = A^{-1}$).
}
\label{fig:ct2}
\end{figure}

Next, still considering the quantile levels $\{0.1, 0.25, 0.5, 0.75, 0.9\}$, we compare the closed testing procedure with the generalized rank-score test as local test and $B=I_5$ to the permutation-based multiple testing procedure proposed by \citet{mrkvicka2023globalquantileregression}. Table~\ref{tab:fwer_perm} reports the empirical FWER under the null model
$y_i \sim \mathcal{N}(0.5 + 0.5 z_i,\; \sigma_i^2)$, where $\sigma_i^2 \in \{1,\; 1+|x_i|,\; 1+|z_i|\}$.
\begin{table}[t]
\centering
\caption{Empirical FWER for the closed testing and the permutation-based \citep{mrkvicka2023globalquantileregression} procedures under the three variance specifications.}
\label{tab:fwer_perm}
\begin{tabular}{l|rr}
\toprule
\textbf{$\sigma_i^2$} & \textbf{Closed Testing} & \textbf{Permutation-based test \citep{mrkvicka2023globalquantileregression}}\\
\midrule
$1$         & 0.036 & 0.039 \\
$1+|x_i|$   & 0.049 & 0.140 \\
$1+|z_i|$   & 0.037 & 0.044 \\
\bottomrule
\end{tabular}
\end{table}
%The proposed approach controls the FWER satisfactorily across all scenarios. By contrast, the permutation-based procedure of \citet{mrkvicka2023globalquantileregression} exhibits inflated FWER when the variance depends on $x_i$. 
The empirical power is then evaluated for 
$y_i \sim \mathcal{N}(0.5 + \beta x_i + 0.5 z_i,\; \sigma_i^2)$, with $\beta \in \{0.4, 0.6\}$ and the same variance specifications. Figure~\ref{fig:perm} reports the empirical power of the individual tests across the five quantile levels. In the settings $\sigma_i^2 = 1$ and $\sigma_i^2 = 1+|z_i|$, where both procedures control the FWER satisfactorily, the closed testing procedure is generally more powerful. The central panel, corresponding to $\sigma_i^2 = 1+|x_i|$, should be interpreted with caution, since in this case the permutation-based procedure of \cite{mrkvicka2023globalquantileregression} does not control the FWER, as shown in Table~\ref{tab:fwer_perm}, and therefore does not provide a valid basis for a power comparison. Since the permutation-based approach fails to provide reliable FWER control when $\sigma_i^2$ depends on $x_i$, and does not offer a power advantage in the other settings, we do not consider it further in what follows.

\begin{figure}
\centering
\includegraphics[width=\textwidth]{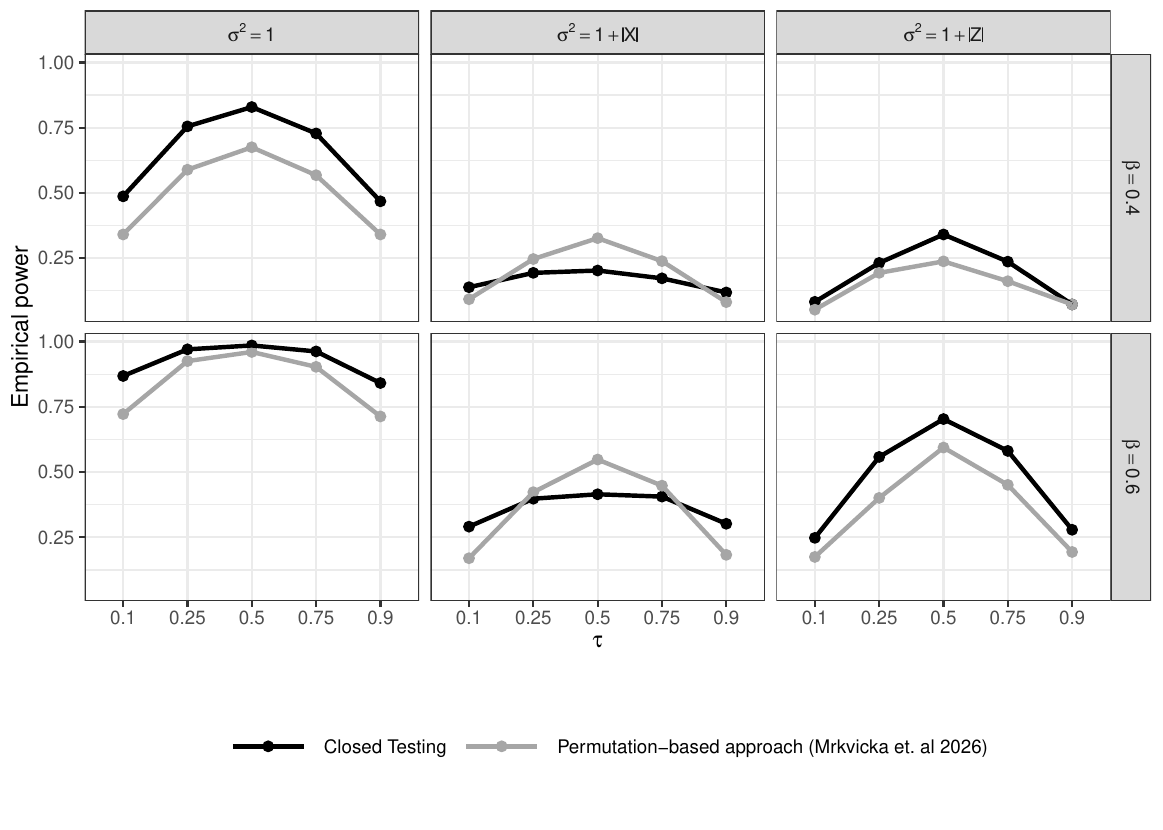}
\caption{Empirical power of the individual tests at quantile levels $\tau \in \{0.1, 0.25, 0.5, 0.75, 0.9\}$ under $\beta \in \{0.4, 0.6\}$. The columns correspond to the three variance settings $\sigma_i^2 \in \{1,\;1+|x_i|,\;1+|z_i|\}$.}
\label{fig:perm}
\end{figure}

We then finally compare the power of the Bonferroni and Holm--Bonferroni corrections with that of the proposed closed-testing procedure.
Three data-generating mechanisms are considered:
%In the first model, we set
%\[
%Y_i = \beta x_i + 0.5\,z_i + \sqrt{3/2}\, w_i,
%\]
%where $w_i$ has a Student's $t$-distribution with five degrees of freedom.
%In the second model, we define 
%\[
%Y_i = w_i - \exp\{[\beta x_i + 0.5\,z_i]_{-30}^{30}\},
%\]
%where $w_i$ follows an exponential distribution with rate $[\beta x_i + 0.5\,z_i]_{-30}^{30}$ and the operator $[\cdot]_{-30}^{30}$ denotes the value of its argument truncated to the interval $[-30,30]$.
%In the third model, we generate
%\[
%Y_i \sim \mathcal{N}\!\bigl(\beta x_i + 0.5\,z_i,\, 1 + |x_i|\bigr).
%\]
\begin{equation*}
\begin{aligned}
\text{(i)}\quad & y_i = 0.5+\beta x_i + {\gamma}z_i + (1+|x_i|)(3/5)^{1/2}\,w_i,\quad w_i \sim t_5,  \\[3pt]
\text{(ii)}\quad & y_i \sim \mathrm{SN}(0.5 + \beta x_i + {\gamma} z_i - 1.453, 3 + |x_i|, 2.2)
,\\[3pt]
\text{(iii)}\quad & y_i \sim \mathcal{N}\left(0.5+ \beta x_i + {\gamma}z_i,\, 1 + |x_i|\right).
\end{aligned}
\end{equation*}
Here $t_5$ is the Student's $t$-distribution with $5$ degrees of freedom, SN denotes the skew-normal distribution, {and $\gamma$ is set to $0.5$.}
%the sample size is set equal to 100 and the covariates $(x_i, z_i)^\top$ are drawn from a bivariate normal distribution with mean zero and an equicorrelated covariance matrix with correlation~0.3.
%The parameter of interest, $\beta$, takes values in $\{0.4, 0.6, 0.8, 1.2\}$.
For any quantile $\tau$, the true coefficient $\beta(\tau)$ depends not only on $\beta$, shared across quantiles, but also on quantile-specific features of the generating distribution. This follows from the definition of the quantile regression coefficient from Section~\ref{setting}.

\begin{figure}
\includegraphics[width=\textwidth]{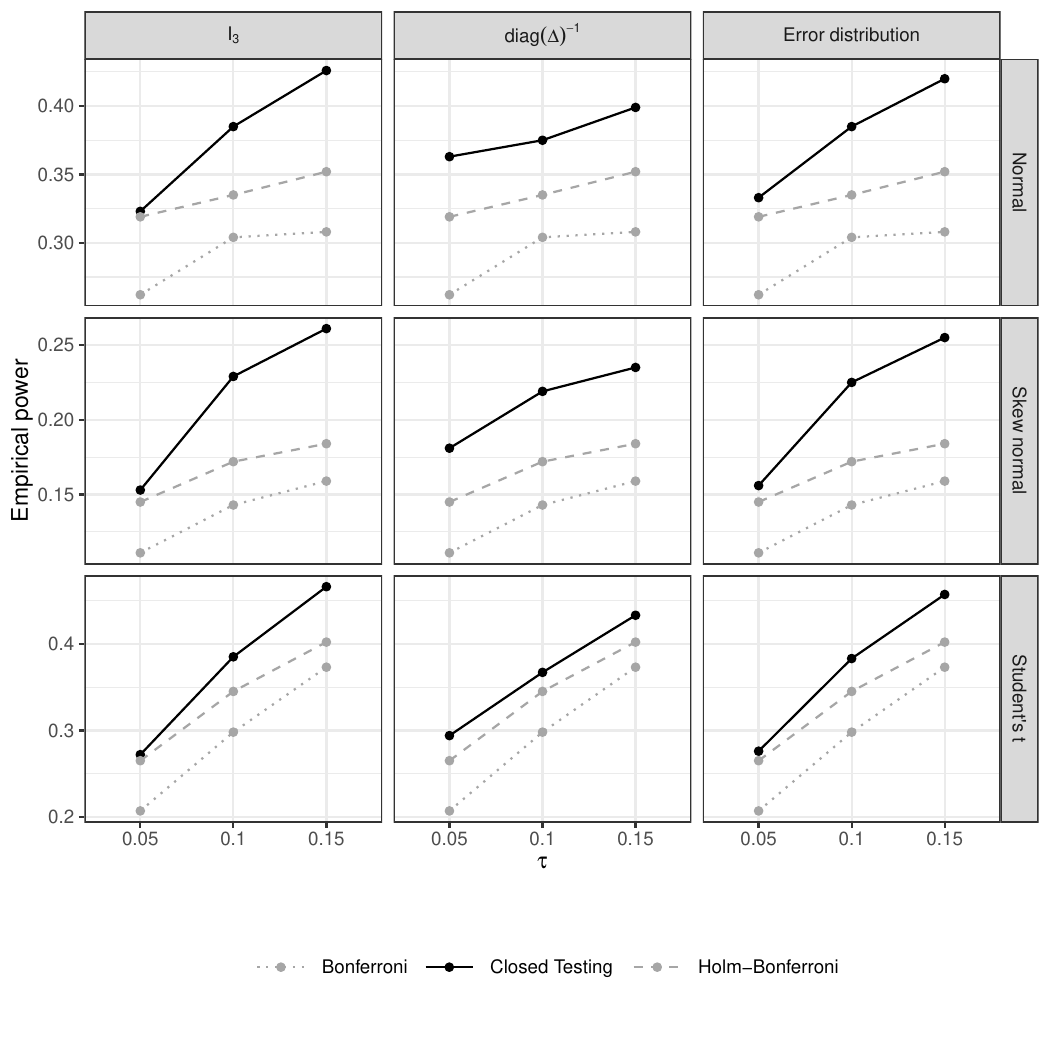}
% note that files may not be rotated
\caption{%Empirical power by the Bonferroni (dotted lines) and closed-testing (solid lines) procedures.
Empirical power for individual hypotheses. Solid black lines represent the closed testing procedure, while dashed and dotted grey lines represent the Holm-Bonferroni and Bonferroni corrections, respectively, all applied to generalized rank-score tests with $B = f \circ F^{-1}(\tau)$, $B=I_{3}$ and $B = \text{diag}\{\Delta\}^{-1}$.
}
\label{fig:power2}
\end{figure}

We compare the power of the proposed closed testing procedure with that of the Bonferroni and Holm–Bonferroni corrections for $\beta = 0.6$, focusing on extreme quantile levels, namely $\{0.05, 0.1, 0.15\}$, and considering three different specifications of $B$. We emphasize that the choice of the weighting matrix must be made prior to observing the data in order to avoid cherry picking. Nevertheless, Fig.~\ref{fig:power2} shows that the identity specification represents a safe choice, even in challenging settings corresponding to regions of the distribution with low density. Analogous results, not reported to avoid repetitions, were obtained for $\gamma = 0$ and $\beta \in \{0.4, 0.8, 1.2\}$, and for correlations between $x_i$ and $z_i$ equal to $0$ and $0.6$, indicating that the findings are robust to changes in nuisance parameters and covariate dependence. 

\begin{figure}
\includegraphics[width=\textwidth]{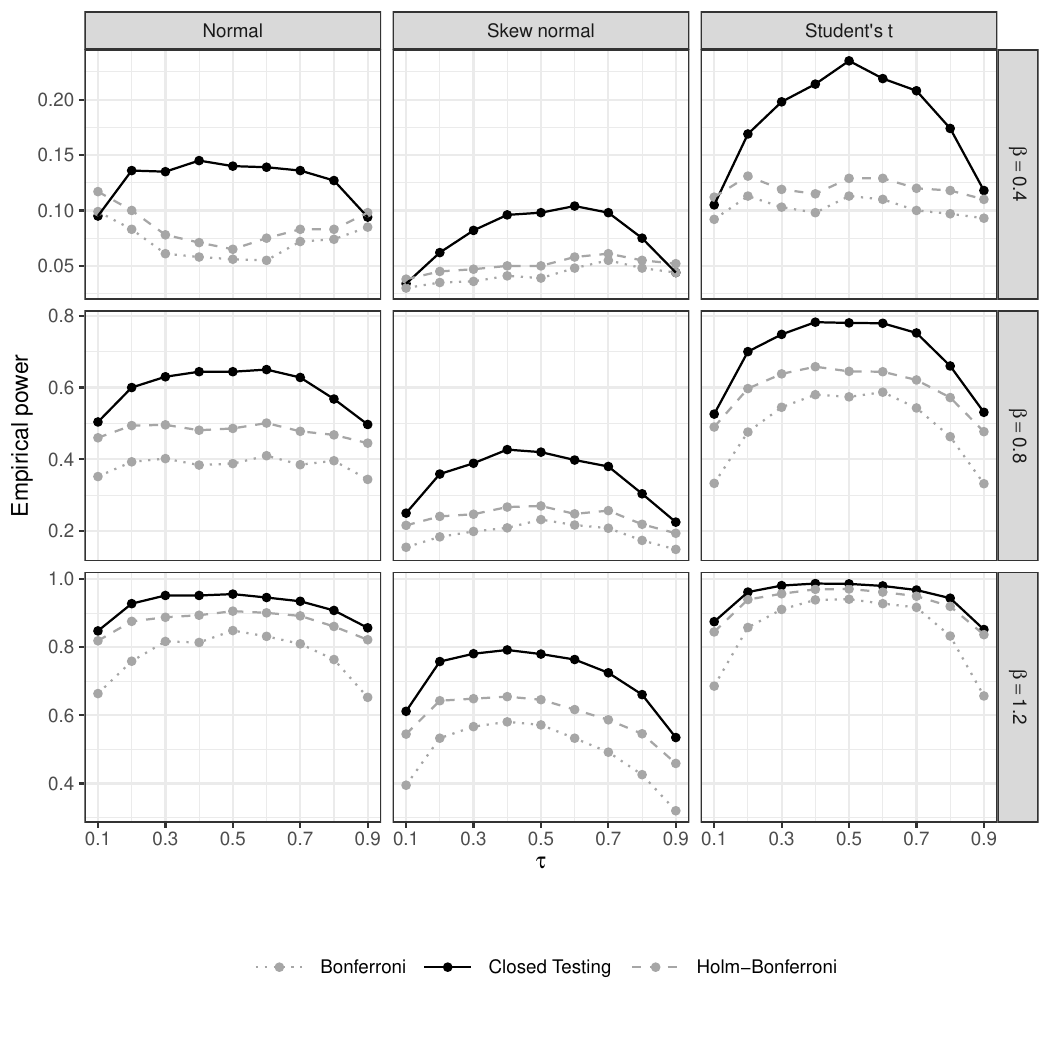}
% note that files may not be rotated
\caption{%Empirical power by the Bonferroni (dotted lines) and closed-testing (solid lines) procedures.
Empirical power for individual hypotheses. Solid black lines represent the closed testing procedure, while the dashed and dotted grey lines represent the Holm-Bonferroni and Bonferroni corrections, respectively, all applied to generalized rank-score tests with $B=I_{9}$.
}
\label{fig:power}
\end{figure}

Finally, the power behavior for $B = I_9$, shown in Fig.~\ref{fig:power}, is evaluated at each decile quantile. The three scenarios reveal several common patterns. Overall, our procedure attains higher power than traditional multiplicity corrections, with only minor exceptions, and displays consistent convergence toward the maximum attainable power as the true parameter increases. As before, analogous results (not reported) were obtained for different values of $\gamma$ and $\beta$, as well as for different correlations between $x_i$ and $z_i$.

\section{Real data application}
We analyze the dataset \texttt{birthwt} \citep{hosmer2013applied} available in the \texttt{R} package \texttt{MASS} \citep{venables2013modern}. The dataset contains 189 observations on infant birth weight together with several potential risk factors for low birth weight. Our main interest is in the effect of maternal smoking during pregnancy, coded as a binary variable (smokers/non-smokers), on birth weight measured in grams. %Further cofounders are the mother's age, the mother's weight at the last menstrual period, the number of previous premature labors, the number of physician visits during the first trimester (continuous variables), the mother's race (three-level factor), history of hypertension, and presence of uterine irritability (two-level factors). 
The analysis adjusts for the following additional covariates: maternal age, maternal weight at the last menstrual period, the number of previous premature labors, and the number of physician visits during the first trimester, all treated as continuous variables; maternal race, included as a three-level factor; and history of hypertension and presence of uterine irritability, both included as binary factors. 
We simultaneously test the significance of the coefficient associated with smoking across the deciles $\{0.1, \dots, 0.9\}$. We report the unadjusted $p$-values, and we compare the multiple testing adjustments of Bonferroni, Holm, and the proposed closed-testing procedure, adopting the identity matrix as weight matrix.

\begin{table}[h]
\centering
\begin{tabular}{l|rrrr}
\toprule
\textbf{$\tau$} & \textbf{Unadjusted} & \textbf{Bonferroni} & \textbf{Holm-Bonferroni} & \textbf{Closed Testing} \\ \midrule
0.1 & 0.573 & 1.000 & 0.573 & 0.573 \\
0.2 & 0.141 & 1.000 & 0.314 & 0.212 \\
0.3 & 0.105 & 0.941 & 0.314 & 0.147 \\
0.4 & 0.007 & 0.059 & 0.039 & 0.030 \\
0.5 & 0.003 & 0.027 & 0.024 & 0.016 \\
0.6 & 0.001 & 0.011 & 0.011 & 0.009 \\
0.7 & 0.004 & 0.033 & 0.026 & 0.018 \\
0.8 & 0.023 & 0.207 & 0.115 & 0.053 \\
0.9 & 0.064 & 0.572 & 0.254 & 0.152 \\ \bottomrule
\end{tabular}
\caption{Unadjusted and multiplicity-adjusted $p$-values for testing the effect of maternal smoking on birth weight across the deciles $\{0.1,\ldots,0.9\}$ in the \texttt{birthwt} dataset.}
\label{tab:realdata}
\end{table}

The results are reported in Table \ref{tab:realdata}. The unadjusted $p$-values are markedly smaller, resulting in more rejections, but they provide no control for multiple comparisons. Among methods that take into account the multiplicity issue, our proposal leads to lower $p$-values (and hence greater statistical power) while maintaining FWER control.

\section{Discussion}\label{discussion}
This manuscript builds on quantile regression models, which are widely applicable and popular across several scientific communities. We construct a framework for safe inference, in terms of error control, when simultaneously testing the effect of a covariate of interest across multiple quantiles. At the same time, the procedure allows for a significant improvement in power over Bonferroni-type corrections without additional assumptions: this translates into additional true findings, that is, significant effects, without inflating the number of false positives. The proposed methodology is general and flexible, allowing the use of different weighting matrices, each with distinct power properties. However, the choice should not be made post-hoc to avoid the phenomenon known as cherry-picking.
%The possibility of choosing different weighting matrices allows for relevant flexibility; however, the choice should not be done post-hoc, to avoid the phenomenon known as cherry picking.

We remark here that the extension to the case of multiple covariates of interest is also conceptually straightforward, requiring only routine matrix calculations; we do not pursue it here in order to avoid additional notational and algebraic complexity.

Future works may proceed in two directions, both related to extensions of the basic quantile regression models: on one hand, the possibility to specify a non-linear relationship between the outcome and the covariates; on the other hand, the introduction of cluster effects and general correlated observations. We believe that future developments in testing and confidence intervals for quantile models should rely on a rank-score-based approach. This framework is significantly more reliable in terms of the convergence of the test statistic to the nominal distribution compared to the more popular Wald-based inference.
%whose reliability -- in terms of convergence of the test statistic to the nominal distribution -- is significantly improved over the more popular Wald-based inference.

\section*{Supplementary material}
\label{SM}
The code used for the simulation study is available at \url{https://github.com/angeella/quasar/tree/main/simulations}
, whereas the code used for the real-data application is available at \url{https://github.com/angeella/quasar/tree/main/application}

\bibliography{bibliography}
\bibliographystyle{apalike}

\end{document}